\documentclass[11pt]{article}

\usepackage{amsthm}
\usepackage{amsfonts}
\usepackage{amsmath}
\usepackage{amssymb}
\usepackage{color}
\usepackage{easybmat}
\usepackage{blkarray}
\usepackage{datetime}
\usepackage{hyperref}

\usepackage[T2A]{fontenc}

\usepackage{mathrsfs}
\renewcommand{\mathcal}[1]{\mathscr{#1}}

\usepackage{marvosym}

\usepackage{theory}
\usepackage{theorems}

\newcommand{\vertex}{\textsc{VertexCover}}
\newcommand{\tropdim}{\textsc{TropDim}}
\newcommand{\tropsolv}{\textsc{TropSolv}}
\newcommand{\tropimpl}{\textsc{TropImpl}}
\newcommand{\tropequiv}{\textsc{TropEquiv}}
\newcommand{\MAP}{\textsc{MAP}}
\newcommand{\ta}{\oplus}
\newcommand{\tp}{\odot}

\newcommand{\bb}[1]{\mathbb{#1}}

\ddmmyyyydate
\title{Complexity of tropical and min-plus\\ linear prevarieties.}

\author{Dima Grigoriev$^1$, Vladimir V. Podolskii$^2$\\[3pt]
$^1$\small CNRS, Math\'ematiques, Universit\'e de Lille, France\\
\small \href{mailto:Dmitry.Grigoryev@math.univ-lille1.fr}{Dmitry.Grigoryev@math.univ-lille1.fr}\\
$^2$ \small Steklov Mathematical Institute, Moscow, Russia\\
 \small   \href{mailto:podolskii@mi.ras.ru}{podolskii@mi.ras.ru}
}

\date{}

\begin{document}

\maketitle

\begin{abstract}
A tropical (or min-plus) semiring is a set $\mathbb{Z}$ (or $\mathbb{Z \cup \{\infty\}}$)
endowed with two operations: $\oplus$, which is just usual minimum, and $\odot$, which is usual addition.
In tropical algebra the vector $x$ is a solution to a polynomial $g_1(x) \oplus g_2(x) \oplus \ldots \oplus g_k(x)$,
where $g_i(x)$'s are tropical monomials, if the minimum in $\min_i(g_{i}(x))$ is attained at least twice.
In min-plus algebra solutions of systems of equations of the form $g_1(x)\oplus \ldots \oplus g_k(x) = h_1(x)\oplus \ldots \oplus h_l(x)$
are studied.

In this paper we consider computational problems related to tropical linear system.
We show that the solvability problem (both over $\mathbb{Z}$ and $\mathbb{Z} \cup \{\infty\}$)
and the problem of deciding the equivalence of two linear systems
(both over $\mathbb{Z}$ and $\mathbb{Z} \cup \{\infty\}$) are equivalent under polynomial-time reduction to mean payoff games
and are also equivalent to analogous problems in min-plus algebra.
In particular, all these problems belong to $\mathsf{NP} \cap \mathsf{coNP}$.
Thus we provide a tight connection of computational aspects of tropical linear algebra with mean payoff games and min-plus linear algebra.
On the other hand we show that computing the dimension of the solution space of a
tropical linear system and of a min-plus linear system are $\mathsf{NP}$-complete.

We also extend some of our results to the systems of min-plus linear inequalities.
\end{abstract}

\section{Introduction}
\label{sec.intro}

A \emph{min-plus} or \emph{tropical semiring} is defined by the set $K$ endowed with two operations
$\ta$ and $\tp$. For $K$ we can take $\mathbb{Z}$, $\mathbb{R}$, $\mathbb{Z} \cup \{+\infty\}$, $\mathbb{R} \cup \{+\infty\}$ and so on.
In this paper we mainly consider the cases of $\bb{Z}$ and $\bb{Z}_{\infty} = \bb{Z} \cup \{+\infty\}$.
Our results also extend to the cases of $\bb{Q}$ and $\bb{Q}_{\infty} = \bb{Q} \cup \{\infty\}$.
The operations \emph{tropical addition} $\ta$ and \emph{tropical multiplication} $\tp$ are defined in
the following way:
$$
x \ta y = \min\{x,y\}, \ \ \ \ x \tp y = x + y.
$$

By the \emph{tropical linear system} associated with a matrix $A \in K^{m \times n}$
we call the system of expressions
\begin{equation} \label{eq.tropical}
\min_{1 \leq j \leq n} \{a_{ij} + x_{j}\},\ 1 \leq i \leq m,
\end{equation}
or to state it the other way the vector $A \tp x$ for $x = (x_{1}, \ldots, x_{n})$.
We say that $x \neq (\infty, \ldots, \infty)$ is a solution to the system~\eqref{eq.tropical} if for every row $1 \leq i \leq m$
there are two columns $1 \leq k < l \leq n$ such that
$$
a_{ik} + x_{k} = a_{il} + x_{l} = \min_{1 \leq j \leq n} \{a_{ij} + x_{j}\}.
$$
Following the notation of~\cite{RGST05first_steps} we call the set of solutions of the tropical linear system by the
\emph{tropical linear prevariety}.
It follows from the analysis of~\cite{RGST05first_steps} that this set
is a union of polyhedrals of possibly different dimensions (this is one of the reasons for using pre- in ``prevariety'').
We call by \emph{the dimension} of the tropical prevariety the largest dimension of the
polyhedron contained in it.

By the \emph{(two sided) min-plus linear system} associated with a pair of matrices $A, B \in K^{m \times n}$
we call the system
\begin{equation} \label{eq.min-plus}
\min_{1 \leq j \leq n} \{a_{ij} + x_{j}\} = \min_{1 \leq j \leq n} \{b_{ij} + x_{j}\},\ 1 \leq i \leq m.
\end{equation}

By the \emph{(two sided) min-plus linear system of inequalities} associated with a pair of matrices $A, B \in K^{m \times n}$
we call the system
\begin{equation} \label{eq.min-plus-ineq}
\min_{1 \leq j \leq n} \{a_{ij} + x_{j}\} \leq \min_{1 \leq j \leq n} \{b_{ij} + x_{j}\},\ 1 \leq i \leq m.
\end{equation}

We note that for all systems we consider it is not essential which of the function to use $\min$ or $\max$.
The whole theory remains the same.

The two branches of algebra related to $(\min, +)$ structure --- tropical algebra and min-plus algebra ---
have different origins. Tropical algebra had arisen in algebraic geometry (see surveys~\cite{itenberg07geometry,sturmfels02equations})
and min-plus algebra had arisen in scheduling theory (see recent monograph~\cite{butkovic10systems}).
Thus the theories in these two branches are different and develop mostly in parallel.
Concerning the computational aspects of these algebras, the most basic question is linear algebra area.
In the case of classical algebra one of the most known Gauss algorithm solves linear systems in polynomial time.
In the case of tropical semiring things turn out to be more complicated and no polynomial time algorithm is known
neither for tropical linear systems, nor for min-plus linear systems.
For the tropical case it is known however that the problem is in $\NP \cap \coNP$,
there are also pseudopolynomial algorithms~\cite{grigoriev10system,akian12mean_payoff}
and also it is known that the problem reduces to the well known and long standing problem
\emph{mean payoff games}~\cite{akian12mean_payoff} (see Section~\ref{sec.prelim} for the definition).
Concerning the algorithms, Grigoriev~\cite{grigoriev10system} has constructed an algorithm which is
pseudopolynomial and at the same time polynomial for constant size matrices, that is
at the same time its running time is bounded by $\poly(m,n) M \log M$ and $\poly(2^{nm},\log M)$, where $n$ is the number of columns,
$m$ is the number of rows, and $M$ is the largest absolute value of matrix entries.
Concerning the dependence on $n$ and $m$ in the second upper bound the best known upper bound
is roughly ${m+n} \choose {n}$ which
was proven by Davydov~\cite{davydov12}.
It was also shown in~\cite{davydov12} that
this is tight upper bound for Grigoriev's algorithm.

More is known about the solvability problem for min-plus linear systems.
In addition to containment in $\NP \cap \coNP$ and pseudopolynomial algorithms, as for tropical systems,
it was proven by Bezem et al.~\cite{Bezem10IPL} that the problem is polynomial-time equivalent to mean payoff games.

One more complexity aspect of min-plus algebra related to our consideration is the solvability problem
of min-plus systems of linear inequalities. In the classical case the corresponding problem is essentially linear programming
which was known for some time to be in $\NP \cap \coNP$ and was proven finally to be in $\P$~\cite{khachijan79}.
Thus intuitively the corresponding problem in min-plus algebra seems to be harder than solving systems of linear min-plus equations
(and one can see that inequalities are formally not harder than equalities in min-plus linear algebra).
For systems of min-plus linear inequalities it is also known that the solvability problem is equivalent to mean payoff games~\cite{akian12mean_payoff}.

The first result of our paper is that the solvability problem for tropical linear systems
is also equivalent to mean payoff games.
Thus on one hand we characterize the complexity of solvability problem of tropical linear systems
and on the other hand give a new reformulation of mean payoff games.
In particular, our result means that the solvability problem for mean payoff games
is equivalent to the solvability problem for min-plus systems. Thus we establish the tight connection
between two branches of algebra over operations $\min$ and $+$.
Also from our reduction the translation of Grigoriev's algorithm to mean payoff games follows.
We are not aware of a ``natural'' algorithm for mean payoff games with the same properties
(of course one can always obtain an ``unnatural'' algorithm from one with the first bound and the other with the second bound performing them in parallel).
This indicates that this translated algorithm might be essentially different from known algorithms for mean payoff games.

Next we study other problems related to tropical linear systems:
the problem of equivalence of two given tropical linear systems and
the problem of computing the dimension of the tropical prevariety.
The former problem turns out to be also equivalent to mean payoff games.
The analogous statement for min-plus linear systems is also true and follows from known result (see Lemma~\ref{lemma.min-plus} below).

Interestingly, the dimension problem of the tropical prevariety
turns out to be $\NP$-complete.
More precisely we prove $\NP$-completeness of the following problem: given an $m \times n$ matrix $A$ and the number $k$
decide whether the dimension of the tropical prevariety of the tropical linear system corresponding to $A$
is at least $k$.
We also prove the analogous result for the case of min-plus linear systems and min-plus systems of inequalities.

All results above we prove for both $\bb{Z}$ and $\bb{Z}_{\infty}$ domains (there is no obvious translation between these two cases).

The techniques of our proofs are mostly combinatorial. For equivalence of solvability problem to mean payoff games
we use the result of~\cite{MSS04scheduling} in which the equivalence of mean payoff games to \emph{max atom problem} ($\MAP$ for short)
was shown (see Section~\ref{sec.prelim} for definitions). This result was already used in~\cite{Bezem10IPL} to show that the solvability problem
of min-plus linear systems is equivalent to mean payoff games.
It was shown there that solvability problem for min-plus linear systems is equivalent to $\MAP$.
For our result we show that solvability problem for tropical linear systems is equivalent to $\MAP$.
The main difficulty here is that $\MAP$ is easier to use for studying min-plus structures than for the tropical ones.
From equivalence of solvability problem to mean payoff games some equivalences between purely tropical
computational problems follows. We also give direct combinatorial proof not referring to mean payoff games
of reductions between these problems.
For dimension problem of tropical linear systems we give a reduction from the vertex cover problem.
The main technical ingredient here is a combinatorial characterization of the dimension of the tropical prevariety
of given tropical linear system.


The rest of the paper is organized as follows.
In Section~\ref{sec.prelim} we state the facts we need on the tropical linear systems.
In Section~\ref{sec.solvability} we prove the result on equivalence of solvability problem
for tropical linear systems and of mean payoff games.
In Section~\ref{sec.rank} we discuss the relation between the dimension of the solution space of the tropical linear system
and the known notions of tropical rank.
In Sections~\ref{sec.characterization} and~\ref{sec.completeness} we prove $\NP$ completeness of the dimension of the tropical prevariety:
in the former we give a combinatorial characterization of the dimension and in the latter we use it to prove $\NP$-completeness.

\section{Preliminaries}
\label{sec.prelim}

Throughout the paper for an integer $n$ we denote by $[n]$ the set $\{1, 2, \ldots, n\}$.
By $\leq_m$ we denote polynomial time many to one reduction.
By $\leq_T$ we denote polynomial time Turing reduction.

\subsection{Mean payoff games}

In an instance of mean payoff game we are given a directed bipartite graph $G = (V, E)$, whose
vertices are divided into two disjoint sets $V =  V_{1} \sqcup V_{2}$,
some fixed initial node $v \in V_{1}$ and a function $w \colon E \to \bb{Z}$ assigning weights to
the edges of $G$. In the beginning of the game a token is placed to the initial vertex $v$. On each turn one
of the players moves the token to some other node of the graph. Each turn of the
game is organized as follows. If the token is currently in some node $u \in V_{1}$ then the first player can move it to any
node $w$ such that $(u,w) \in E$. If, on the other hand, $u \in V_{2}$ then the second player can move the token to any node
$w$ such that $(u,w) \in E$. The game is infinite and the process of the game can be described by the sequence of nodes
$v_{0}, v_{1}, v_{2}, \ldots$ which the token visits. Note that $v_0 = v$. The first player wins the game if
\begin{equation} \label{eq.mpg}
\liminf_{n \to \infty} \frac 1t \sum_{i=1}^{t} w(v_{i-1}, v_{i}) > 0.
\end{equation}
The corresponding mean payoff game problem is to decide whether the first player has winning strategy.

For more information on mean payoff games see survey~\cite{klauck02parity_games}.
It is known that both of the players have optimal positional strategy,
that is strategies depending only on the current position of the token and not on the history.
From this in particular it follows that the optimal value of the game (the largest left-hand side of~\eqref{eq.mpg} that the first player can achieve)
is a rational number with the denominator polynomial in the number of vertices of $G$.

Also it is clear that the instance of negated mean payoff game problem, that is
the problem whether the second player has a winning strategy, is polynomial time m-reducible to mean payoff games.
Indeed, just change the roles of the players and add the new initial vertex $v^\prime$ with no ingoing edges and one outgoing
edge $(v^{\prime}, v)$ to pass the move to the second player. The problem that the value of the game might be zero can be handled
by changing all weights by small rational number (after that the value of the game is always nonzero) and multiplying them
by the denominator to make them integer.

During our reductions sometimes we will be in the situations when we reduce some problem to
solution of several instances of another problem equivalent to mean payoff games, that is
the input to the original problem will be `yes' instance iff all inputs constructed during the reduction
are `yes' inputs of the problem equivalent to mean payoff games. In this case we can actually substitute several inputs by one
since we can do this for mean payoff games. Indeed, we can just consider the graph consisting of unconnected copies
of all graphs corresponding to several inputs we have, add the node belonging to the second player
from which he can reach all starting nodes of all subgraphs and add one more node to pass the first move to the first player.

\subsection{Tropical and min-plus linear systems}

Consider arbitrary tropical linear system~\eqref{eq.tropical}.
Note that its tropical prevariety $S$ is closed under tropical scalar multiplication,
or, to state it the other way, $S = S + \mathbb{Z} \vec{1}$, where by $\vec{1}$ we denote the vector of all ones.
Thus we can consider the set of solutions of~\eqref{eq.tropical} as a set in the tropical projective space $\mathbb{TP}^{n-1} = \mathbb{R}^n / \langle \vec{1} \rangle_{\mathbb{R}}$. In this paper we will alternatively consider the solution prevariety in the spaces $\bb{R}^n$ and $\bb{TR}^{n-1}$
depending on which one is more convenient in the current argument.

Consider some matrix $A \in \bb{Z}^{m \times n}$.
Note that adding some number to all entries of some row of $A$ does not change the tropical prevariety
of system~\eqref{eq.tropical}.
Thus in the course of the proofs we can freely add and subtract some number from some row of the matrix under consideration.

Let us add the same vector $\vec{v} \in \bb{Z}^n$ to all rows of $A$ and denote the resulting matrix by $A_{\vec{v}}$.
Then we have that the tropical prevariety of $A_{\vec{v}}$ is a linear translation
of the tropical prevariety of $A$. Since many important properties survive after translations we will apply this kind of transformation
to matrices.

Finally, let us multiply all entries of the matrix by the same constant $c \in \bb{N}$.
Note that all vectors in the tropical prevariety also multiplies by the same constant.
Sometimes we will perform this operation also.
In particular,
this observation implies that all our results are also true for the domains $\bb{Q}$ and $\bb{Q} \cup \{\infty\}$.

All observations above are also true for min-plus systems of equalities and inequalities.

Consider a tropical linear system with the matrix $A \in \bb{Z}^{m \times n}$ and assume that
$a_{ij} \geq 0$ for all $i \in [m],\ j \in [n]$ (we can reduce any matrix to such form adding vectors $c \cdot \vec{1}$ to the rows).
Assume that the entries of the matrix are bounded by some value $M$, that is $a_{ij} \leq M$.

The following lemma bounding the size of the smallest solution was proven in~\cite{grigoriev10system}.

\begin{lemma} [\cite{grigoriev10system}]
If the system has a solution $(x_{1}, \ldots, x_{n})$, then it has a solution $(x^{\prime}_{1}, \ldots, x^{\prime}_{n})$
satisfying $0 \leq x^{\prime}_{j} \leq M$
for all $1 \leq j \leq n$.
\end{lemma}

It is known that the solution space of $A$ is the system of polytopes of possibly different dimension~\cite{RGST05first_steps}.
It is also known that the solution space is connected (see~\cite{theobald06frontiers}, Lemma 4.12).

In this paper we consider the following problems.

\begin{itemize}
  \item $\tropsolv$. In this problem we are given an integer matrix $A \in \bb{Z}^{m \times n}$.
  The problem is to decide whether the corresponding tropical system~\eqref{eq.tropical} is solvable.
  \item $\tropequiv$. In this problem we are given two integer matrices $A \in \bb{Z}^{m \times n}$ and $B \in \bb{Z}^{k \times n}$.
  The problem is to decide whether the corresponding tropical systems~\eqref{eq.tropical} over the same set of variables are equivalent.
  \item $\tropimpl$. In this problem we are given an integer matrix $A \in \bb{Z}^{m \times n}$ and a vector $l \in \bb{Z}^{n}$.
  The problem is to decide whether the tropical system~\eqref{eq.tropical} corresponding to $A$ implies the tropical equality corresponding to $l$.
  \item $\tropdim$. In this problem we are given an integer matrix $A \in \bb{Z}^{m \times n}$ and a number $k \in \bb{N}$.
  The problem is to decide whether the dimension of the tropical prevariety corresponding to the tropical system~\eqref{eq.tropical} is at least $k$.
\end{itemize}
For all problems above there are also variants of them over $\bb{Z}_{\infty}$. We denote them by the subscript $\infty$,
for example in the problem $\tropsolv_{\infty}$ we are given a matrix $A \in \bb{Z}_{\infty}^{m \times n}$ and
the problem is to decide whether the corresponding tropical system over $\bb{Z}_{\infty}$ is solvable.
For local dimension of tropical prevariety (that is the dimension of the neighborhood of some point) over $\bb{Z}_{\infty}$
in a point with some infinite coordinates we consider just the dimension over finite coordinates only.

When we consider systems over $\bb{Z}_{\infty}$ we do not allow solutions consisting only of infinities.

Next we show some simple relations between $\bb{Z}$ and $\bb{Z}_{\infty}$ cases.
\begin{lemma} \label{lemma.reduct_simple}
\begin{enumerate}
  \item $\tropsolv \leq_m \tropsolv_{\infty}$;
  \item $\tropimpl \leq_m \tropimpl_{\infty}$;
  \item $\tropdim \leq \tropdim_{\infty}$.
 \end{enumerate}
\end{lemma}
\begin{proof}
For the first reduction, if we are given a tropical linear system with coefficients
in $\bb{Z}$ then it is solvable over $\bb{Z}$ iff it is solvable over $\bb{Z}_{\infty}$.
For the nontrivial direction, if there is a solution over $\bb{Z}_{\infty}$ in which some coordinates are infinite,
we can just substitute them by large enough finite numbers.

For the second reduction, if we are given a tropical linear system and a tropical linear equality over $\bb{Z}$,
consider them over $\bb{Z}_{\infty}$. If there was no implication over $\bb{Z}$, that is there is a solution over $\bb{Z}$ of the system,
which is not a solution of the equation, then clearly the same is is true over $\bb{Z}_{\infty}$, and there is also no implication.
If there is no implication over $\bb{Z}_{\infty}$ then there is a solution over $\bb{Z}_{\infty}$ of the system, which is
not a solution of the equation. Substituting infinities in the solution by large enough constants we get that there is
also no implication over $\bb{Z}$.

For the last reduction, again if we have a tropical linear system with coefficients in $\bb{Z}$
and we have some solution with infinite coordinates then if we substitute infinities by large enough finite numbers,
the local dimension at this point does not decrease.
\end{proof}


\subsection{Max-atom problem}
\label{subsec.max-atom}

For the proof of our first result we need an intermediate
\emph{max-atom problem} or $\MAP$. This problem consists of solving
a system of inequalities of the form
\begin{equation} \label{eq.map}
\max\{x,y\} + k \geq z
\end{equation}
over $\bb{Z}$ where $k$ is also an integer.


\section{Solving tropical systems is equivalent to mean payoff games}
\label{sec.solvability}

In this section we prove that solvability problem for tropical linear systems is equivalent to mean payoff games.
For this we show that $\tropsolv$ is equivalent to $\MAP$.
First we prove the following simple lemma.

\begin{lemma} \label{lemma.inequalities}
$\tropsolv$ reduces in polynomial time to the solvability problem for the system of min-plus inequalities.
Moreover, for given tropical linear system we can effectively construct the system of min-plus inequalities
over the same set of variables and with the same set of solutions.
The same is true for the domain $\bb{Z}_{\infty}$.
\end{lemma}

\begin{proof}
Let $A$ be some tropical linear system.
For each its equation we construct the system of min-plus inequalities
over the same set of variables which is \emph{equivalent} to the equation.

For this let
\begin{equation} \label{eq.row}
\min\{x_1 + a_1, x_2 + a_2, \ldots, x_{n} + a_n\}
\end{equation}
be one of the rows of the system $A$.
For notation simplicity we denote $y_i = x_i + a_i$ for $i = 1, \ldots, n$.
Then we can rewrite~\eqref{eq.row} as $\min\{y_{1}, \ldots, y_{n}\}$.

It is easy to see that the fact that the minimum in the expression above is attained at
least twice is equivalent to the fact that for any $i = 1, \ldots, n$
it is true that
\begin{equation} \label{eq.to_MAP}
\min\{y_1, ..., y_{i-1}, y_{i+1}, ..., y_{n}\} \leq y_{i}.
\end{equation}
And each of these inequalities is in turn equivalent to the inequality
\begin{align*}
& \min\{y_1, ..., y_{i-1}, y_{i} - 1, y_{i+1}, ..., y_{n}\}
\leq \\
& \min\{y_1 - 1, ..., y_{i-1} - 1, y_{i}, y_{i+1} - 1, ..., y_{n} - 1\}.
\end{align*}
The last inequality is already in min-plus form and thus we have that any tropical equality is equivalent
to the system of min-plus inequalities. To get the system of inequalities equivalent to the system of equalities
we just unite systems for all equalities of $A$.

Note that exactly the same analysis works for the case $\bb{Z}_{\infty}$.
\end{proof}

\begin{remark}
It was proven by Akian et al.~\cite{akian12mean_payoff} that the solvability problem for the systems of min-plus inequalities
(over $\bb{Z}$ and $\bb{Z}_{\infty}$) is equivalent to mean payoff
games. It was also proven there that $\tropsolv$ and $\tropsolv_{\infty}$ reduces to mean payoff games. The lemma above shows, in particular, that the latter
result follows easily from the former.
\end{remark}

As a corollary of Lemma~\ref{lemma.inequalities} we have a reduction from $\tropsolv$ to $\MAP$.
\begin{corollary} \label{cor.MAP}
$\tropsolv \leq_{m} \MAP$.
\end{corollary}
\begin{proof}
Given a tropical linear system $A$ first for each equality construct the system of inequalities~\eqref{eq.to_MAP}.
Then multiply all these inequalities by $(-1)$ and make a transformation of variables $x \mapsto -x$ to switch from $\min$
to $\max$.
After that the inequalities are almost in the form of $\MAP$ and can be easily transformed to the desired form
by simple tricks described in Section~2 of~\cite{Bezem10IPL}.
\end{proof}

Now we proceed to the reduction in the reverse direction.
For this we will need the following technical lemma.

\begin{lemma} \label{lemma.stars_restriction}
Let $k \leq n$ and consider arbitrary vector $\vec{a} = (a_1, \ldots, a_{k}) \in \bb{Z}^{k}$.
Then for any $C \in \bb{Z}$ there is a tropical linear system $A \in \bb{Z}^{m \times n}$,
where $m= n-k+1$,
such that
\begin{itemize}
\item for any $i \in [m]$ and any $j \in [k]$ we have $a_{ij} = a_{j}$;
\item for any $i \in [m]$ and any $j \in [n] \setminus [k]$ we have $a_{ij} \geq C$;
\item for any solution of $A$ and for any row the minimum is attained at least twice in the $\vec{a}$-part of the row.
\end{itemize}
\end{lemma}
\begin{proof}
To prove the lemma we will introduce several tropical equations and the system $A$ will be the union of them.
First consider the row corresponding to the following vector
$$
l = (\vec{a}, C+1, \ldots, C+1),
$$
where $l \in \bb{Z}^n$.
Next, for each $i = k+1, \ldots, n$ let
$$
l_i = l - e_{i} = (\vec{a}, C+1, \ldots, C, \ldots, C+1),
$$
where $e_i \in \bb{Z}^n$ is a vector with $1$ in the $i$-th coordinate and $0$ in all other coordinates,
and let
$$
l_0 = l - \sum_{i=1}^{k} e_{i} = (\vec{a} - \vec{1}, C+1, \ldots, C+1).
$$
We let $A$ be the system consisting of equalities $l_0, l_{k+1}, \ldots, l_{n}$.

Suppose, by a way of contradiction, that $A$ has a solution such that in some row $l_i$ there is at most one minimum in the $\vec{a}$-part.
This means that in this row there is a minimum in the column $j$ such that $k+1 \leq j \leq n$.
If $j \neq i$ consider the row $l_j$. It is easy to see that this row contains exactly one minimum (in the column $j$)
and this is the contradiction. Thus the minimum in the row $l_i$ outside of $\vec{a}$-part can be situated only in the column $i$
(in particular, $i \neq 0$).
But since the minimum is attained at least twice there is at least one minimum in $\vec{a}$-part of $l_i$.
Now consider the row $l_0$. Clearly the minimums of this row are the minimums of $\vec{a}$-part of $l_i$
and thus there are at least two of them.
\end{proof}

To prove the desired reduction we will make use of the following lemma
bounding the size of the minimal solution of $\MAP$ which was proven in~\cite{Bezem10IPL}.

\begin{lemma}[\cite{Bezem10IPL}] \label{lemma.bezem}
Let $M$ be a $\MAP$ system over variables $x_{1}, \ldots, x_{n}$ and let
$C$ be the sum of absolute values of all constants in $M$.
Then if $M$ is solvable then it has a solution $\vec{x}$ such that
$\max_{i \in [n]}\{x_{i}\} - \min_{i \in [n]}\{x_i\} \leq C$.
\end{lemma}

Now we are ready to prove the reduction in the backwards direction.

\begin{theorem} \label{thm.MAP}
$\MAP \leq_m \tropsolv$.
\end{theorem}
\begin{proof}
Suppose we are given a system $A$ of inequalities of the form
$\max\{x,y\} + k \geq z$.
First multiply all inequalities by $(-1)$ and make a transformation of variables
$x \mapsto (-x)$. Then we have a system $B$ of inequalities of the form $\min\{x,y\} + k \leq z$
which is solvable if and only if the initial system is solvable.
We denote by $C$ the sum of absolute values of all constants in $B$.

Now we are ready to construct a tropical linear system $T$.
Let us denote variables of $B$ by $x_{1}, \ldots, x_{n}$.
Our tropical linear system for each variable $x_{i}$ of $B$
will have two corresponding variables $x_{i}$ and $x_{i}^{\prime}$.
We would like these variables to be equal in any solution of $T$.
This can be easily achieved by the means of Lemma~\ref{lemma.stars_restriction}.
For this let in this lemma $k=2$, $\vec{a} = (0,0)$, $C = C$ and apply it to the variables $x_{i}, x_{i}^{\prime}$.
As a result we get the system $T_i$ which guaranties that in each its solution variables $x_{i}$ and $x_{i}^{\prime}$
are equal. We include systems $T_{i}$ for all $i$ into the system $T$.

Next we have to guarantee that for any inequality $\min\{x,y\} + k \leq z$ of $B$,
where $x,y,z$ are some variables among $x_{1}, \ldots, x_{n}$,
the same inequality is true for the solutions of $T$.
Since we already know that the variables $x_{i}$ and $x_{i}^{\prime}$
are equal for each solution of $T$, it suffices to say that
$$
\min\{x, x^{\prime}, y, y^{\prime}, z - k, z^{\prime} - k + 1\}
$$
is attained at least twice.
However, we have to add other variables into this inequality.
This can be done again by Lemma~\ref{lemma.stars_restriction}.
For this let in this lemma $k=6$, $\vec{a} = (0, 0, 0, 0, -k, -k+1)$, $C = C$
apply the lemma to the variables $x, x^{\prime}, y, y^{\prime}, z, z^{\prime}$
and include the resulting system to the system $T$.

Now the construction of $T$ is finished and we have to show that it is solvable if and only if $B$ is solvable.
Assume first that $T$ has a solution. Then it follows from the construction of $T$ that
for each $i = 1, \ldots, n$ variables $x_{i}$ and $x_{i}^{\prime}$ are equal.
And from this and again from the construction of $T$ it follows that each inequality of $B$ is true.

On the other hand, suppose that $B$ is satisfiable. Then, by Lemma~\ref{lemma.bezem} there is a solution $\vec{x}$
such that
$$
\max_{i \in [n]}\{x_{i}\} - \min_{i \in [n]}\{x_i\} \leq C.
$$
Since we can add any constant to all coordinates of $\vec{x}$
we can assume that $\min_{i \in [n]}\{x_i\} = 0$ and thus for all $i$ we have $0 \leq x_{i} \leq C$.
For the solution of $T$ let $x_{i}$ be the same as in the solution of $B$ and let $x_{i}^{\prime} = x_{i}$ for all $i$.
It is left to check that this vector is a solution of $T$.
We can check it for all rows separately. If the row is in $T_i$ for some $i$ then clearly the minimum is attained on $x_{i}$ and $x_{i}^{\prime}$
due to the choice of the constant $C$ in application of Lemma~\ref{lemma.stars_restriction}.
And if the row came from some inequality $\min\{x,y\} + k \leq z$ of $B$ then clearly the minimum is attained
either on $x$ and $x^{\prime}$, or on $y$ and $y^{\prime}$.
\end{proof}

From Theorem~\ref{thm.MAP} and Corollary~\ref{cor.MAP} we conclude the following.
\begin{corollary}
The problem $\tropsolv$ is polynomially equivalent to mean payoff games.
\end{corollary}

Moreover we can also conclude the same for the problem $\tropsolv_{\infty}$.
\begin{corollary}
The problem $\tropsolv_{\infty}$ is polynomially equivalent to mean payoff games.
\end{corollary}
\begin{proof}
It was proven in Akian et al.~\cite{akian12mean_payoff} that $\tropsolv_{\infty}$ is polynomial time reducible to mean payoff games
(see also the remark after Lemma~\ref{lemma.inequalities}).
Theorem~\ref{thm.MAP} gives us that mean payoff games can be reduced to $\tropsolv$. Finally,
$\tropsolv$ reduces to $\tropsolv_{\infty}$ by Lemma~\ref{lemma.reduct_simple} and thus all three problems are equivalent.
\end{proof}
In particular, it follows that the problems $\tropsolv$ and $\tropsolv_{\infty}$ are polynomial time equivalent.
But the given proof of equivalence of these two purely tropical problems rather unnaturally goes through mean payoff games.
In Appendix~\ref{app.solv_direct} we give a direct proof of this equivalence.

One more corollary of our analysis concerns the equivalence and implication problems for tropical linear systems.
\begin{corollary} \label{cor.equiv}
The problems $\tropequiv$, $\tropequiv_{\infty}$ are polynomial time equivalent
to mean payoff games.
The problems $\tropimpl$ and $\tropimpl_{\infty}$ are polynomial time equivalent to mean payoff games under Turing reductions.
\end{corollary}
\begin{proof}
It is easy to see that the problem $\tropequiv$ is equivalent to the problem $\tropimpl$
(under Turing reduction).
Suppose we are given a tropical system $A$ and a tropical equation $l$.
Deciding whether $l$ follows from $A$ is equivalent to deciding whether systems $A$ and $A \cup \{l\}$ are equivalent.
On the other hand, if we need to check whether two systems $A$ and $B$ are equivalent it is enough to check whether
each equation of the second system follows from the first system and vise versa.
Thus we have that $\tropequiv$ is equivalent to $\tropimpl$.
The same argument gives us also that $\tropequiv_{\infty}$ is equivalent to $\tropimpl_{\infty}$.
Note, that the same argument works also for min-plus systems and systems of min-plus inequalities.

Next, it is easy to construct the reduction from $\tropsolv$ to $\tropequiv$.
Indeed, to check whether some system is solvable it is enough to check whether it is equivalent to
some fixed nonsolvable system.

Reduction of $\tropimpl$ to $\tropimpl_{\infty}$ is proven in Lemma~\ref{lemma.reduct_simple}.

Thus it is only left to show that $\tropequiv_{\infty}$ reduces to mean payoff games.
Assume that we are given two tropical systems $A_{1}$ and $A_{2}$ and we have to check whether they are equivalent.
First by Lemma~\ref{lemma.inequalities} for each of the systems we construct the system of inequalities with the same
solution sets. Then we reduce the equivalence problem for the systems of inequalities to implication problem for inequalities
by the same argument as above. And finally we can apply the result of Allamigeon et al.~\cite{AGK11polar_cones} stating that
the implication problem for min-plus inequalities over $\bb{Z}_{\infty}$ is equivalent to mean payoff games.

Keeping in mind the discussion in Preliminaries it is easy to see that these reductions can be transformed into
$m$-reductions for the case of the problems $\tropequiv$ and $\tropequiv_{\infty}$.

\end{proof}

It is not hard to see than analogous results for min-plus linear systems follows along the same lines from known results.
\begin{lemma} \label{lemma.min-plus}
The equivalence and implication problems for min-plus systems of linear equations over both $\bb{Z}$ and $\bb{Z}_{\infty}$
are equivalent to mean payoff games. The same is true for min-plus systems of linear inequalities.
\end{lemma}
\begin{proof}
The same proof as for Corollary~\ref{cor.equiv} works. Instead of Lemma~\ref{lemma.inequalities} we can apply
trivial relation between min-plus equalities and inequalities.

The result on implication problem for min-plus systems of linear inequalities over $\bb{Z}_{\infty}$ was already proven in~\cite{AGK11polar_cones}.
\end{proof}

For both min-plus and tropical linear systems we give a direct combinatorial proofs of equivalence between
solvability and equivalence problems in Appendix~\ref{app.equivalence}.

\section{Dimension and the tropical rank}
\label{sec.rank}

In the case of classical linear systems the dimension of the solution space is closely related to the rank of the matrix.
The natural idea is that maybe the dimension of the tropical
prevariety is also related to some ``rank'' of the tropical matrix
and $\NP$-completeness can be derived from the completeness for this ``rank''.

There are three notions of the ``rank'' in tropical algebra studied in the literature:
\emph{Barvinok rank}, \emph{Kapranov rank} and \emph{tropical rank} (see~\cite{DSS05rank} for the definitions).
For them there is a relation
\begin{equation} \label{eq.ranks}
tropical\ rank(A) \leq Kapranov\ rank(A) \leq Barvinok\ rank(A).
\end{equation}
for any matrix $A$. All inequalities can be strict in~\eqref{eq.ranks}.

We will show the following result.
\begin{lemma}
For any matrix $A \in \mathbf{R}^{m \times n}$ we have
$$
n - tropical\ dimension(A) \leq tropical\ rank(A),
$$
and the inequality can be both tight and strict.
Here by the tropical dimension we mean the affine variant of dimension.
\end{lemma}

This lemma together with~\eqref{eq.ranks} shows that there is a relation between the tropical dimension
and ranks of the tropical matrix, but this relation is not enough for computational needs.

\begin{proof}[Proof of the lemma]

To prove the inequality let the tropical rank of the matrix $A$ be equal to $r$ and consider
the maximal set $C$ of tropically independent columns in $A$,
that is the maximal set of columns such that the tropical linear system generated by them is unsolvable.
The size of this set of columns is equal to $r$ (see~\cite{DSS05rank,izhakian09rank,grigoriev10system}).
Add one of the remaining $n-r$ columns to $C$ and denote the resulting $m \times (r+1)$ matrix by $C^{\prime}$.
The columns in $C^{\prime}$ are tropically dependent, so there is a solution
to the tropical linear system with the columns $C^{\prime}$.
This solution can be extended to the solution of the whole system by fixing all coordinates $x_{i}$
with $i \in [n]\setminus C^{\prime}$ to be large enough numbers.
Note that these coordinates of the resulting solution of $A$ can be changed locally (if the numbers were chosen large enough).
Thus we have that the solution space contains subspace of dimension $n-(r+1)$.
But note that currently we have projective dimension: some of the coordinates never change in this subspace.
So we can add the vector $(1, \ldots, 1)$ to our subspace and get the desired subspace of dimension $n-r$.

To show that the inequality can be tight consider for example the matrix
$$
\left( \begin{BMAT}(e){ccc}{cc}
1 & 0 & 0 \\
0 & 1 & 0 \\
\end{BMAT}
\right).
$$
It is easy to see that the solution space of the corresponding tropical system consists of points
$(c,c,c)$ for any $c$ and thus has dimension $1$. The tropical rank of this matrix is $2$.
To see this consider the submatrix defined by the first two columns.

To show that on the other hand the inequality can be strict consider the matrix
$$
\left( \begin{BMAT}(e){ccccc}{cccc}
0 & 0 & 0 & 0& 0\\
1 & 0 & 0 & 0& 0\\
1 & 1 & 0 & 0& 0\\
1 & 1 & 1 & 0& 0\\
\end{BMAT}
\right).
$$
The tropical rank of this matrix is $4$. For this consider the submatrix defined by the first four columns.
On the other hand the dimension of the solution space is also $4$ since it contains subspace generated by $(1,0,0,0,0)$,
$(0,1,0,0,0)$, $(0,0,1,0,0)$, $(1,1,1,1,1)$.

Both of the examples above can be easily generalized to arbitrary matrix size.
\end{proof}

\section{Combinatorial characterization of the dimension of the tropical prevariety}
\label{sec.characterization}

In our analysis it will be convenient to use the following definition.

\begin{definition}
Let $A$ be a matrix of size $m \times n$. We associate with it the table $A^{*}$ of the same size
$m \times n$ in which we put the star $*$ to the entry $(i,j)$ iff $a_{ij} = \min_{k} a_{ik}$ and
we leave all other entries empty.
\end{definition}

The table $A^{*}$ captures properties of the tropical system $A$ essential to us.
For example, the vector $x = (x_{1}, \ldots, x_{n})$ is a solution to the system $A$ iff
there are at least two stars in every row of the table $(\{a_{ij} + x_{j}\}_{ij})^{*}$.

Next we give a combinatorial characterization of local dimension (at a given point) of a tropical prevariety
in terms of the table $A^*$.
For this we will use the following block-triangular form of the matrix.

\begin{definition}
The block triangular form of size $d$ of the matrix $A$ is a partition of the set of rows of $A$ into sets
$R_{1}, R_{2}, \ldots, R_{d}$ (some of the sets $R_{i}$ might be empty)
and a partition of the set of columns of $A$ into nonempty sets $C_{1}, \ldots, C_{d}$ with the following properties (see figure):
\begin{enumerate}
\item for every $i$ each row in $R_{i}$ has at least two stars in columns $C_{i}$ in $A^*$;
\item if $1 \leq i < j \leq d$ then rows in $R_{i}$ have no stars in columns $C_{j}$ in $A^*$.
\end{enumerate}
\end{definition}

\begin{figure}
\[
A^{\ast} = 
\begin{blockarray}{*{13}{p{4pt}}}
&      & $C_1$  &      &      &   $C_2$        &       &      &  $\cdots$   &       &      &    $C_d$    &     \\
\begin{block}{p{10pt}(*{3}{*{3}{p{4pt}}|}*{3}{p{4pt}})}
       & $\ast$ & $\ast$ &      &      &           &       &      &           &       &      &           &     \\
$R_1$    & $\ast$ &      & $\ast$ &      & $\emptyset$ &       &      & $\emptyset$ &       &      & $\emptyset$ &     \\
       &      & $\ast$ & $\ast$ &      &           &       &      &           &       &      &           &     \\ \cline{2-13}
       & $\ast$ &      &      & $\ast$ & $\ast$      &       &      &           &       &      &           &     \\
$R_2$    & $\ast$ & $\ast$ &      &      & $\ast$      &$\ast$ &      & $\emptyset$ &       &      & $\emptyset$ &     \\
       &      &      &      & $\ast$ &           &$\ast$ &      &           &       &      &           &     \\ \cline{2-13}
       & $\ast$ &      &      &      &           &       &$\ast$&           &$\ast$ &      &           &     \\
$\vdots$ & $\ast$ &      &      & $\ast$ & $\ast$      &$\ast$ &$\ast$&$\ast$     &       &      & $\emptyset$ &     \\
       &      & $\ast$ &      &      &           &       &      &$\ast$     &$\ast$ &      &           &     \\ \cline{2-13}
       &      &      &      &      & $\ast$      &       &      &           &       &      &$\ast$     & $\ast$ \\
$R_d$    &      &      & $\ast$ &      &           &$\ast$ &      &$\ast$     &       &$\ast$&$\ast$     &     \\
       &$\ast$&$\ast$&      &      &           &       &$\ast$&           &       &$\ast$&$\ast$     &     \\
\end{block}
\end{blockarray}
\]
\end{figure}

We are looking for a block triangular form of the matrix with the largest possible $d$.
We next make several observations on the structure of the block triangular form
of the maximal size:
\begin{itemize}
  \item Without loss of generality the pairs $(C_{i}, R_{i})$ with empty $R_{i}$ can be moved to the beginning of the list
  permuting correspondingly the list of $C_i$-s and the list of $R_i$-s.
  \item We can assume that the pairs with empty $R_{i}$ have $|C_{i}|=1$. Indeed, if $|C_{i}|>1$ we can break
  it into several sets of size $1$ without violating the properties of the block triangular form and the size will only increase.
\end{itemize}

Now we are ready to give a combinatorial characterization of the dimension of the prevariety.

\begin{theorem} \label{thm.dim_char}
Assume that the zero vector is a solution of the tropical linear system $A$.
Then the local projective dimension of the system $A$ in zero solution is equal to the maximal $d$ such that there is
a block triangular form of $A$ of size $d+1$.
\end{theorem}

Clearly the case of arbitrary solution can be reduced to the zero solution.
\begin{proof}
We also can assume that the minimum in each row of $A$ is $0$.

Denote by $d$ the size of the largest block triangular form minus $1$
and denote the local projective dimension of the tropical prevariety in zero solution by $\dim A$.

It is not hard to see that $\dim A \geq d$.
Indeed, consider the block triangular form of size $d+1$ and take as the basis for the subspace of tropical prevariety
the negated characteristic vectors of $\cup_{j \geq i} C_j$ for all $i = 2, \ldots, d+1$.
It is clear that any point in this space close enough to the zero vector is a solution to the tropical system.

It remains to prove that $\dim A \leq d$.
Consider the polytope of the largest dimension in the tropical prevariety containing the zero point.
We can restrict ourselves to a cone of the same dimension which
vertex is zero point and such that the neighborhood of the vertex intersected with the cone lies in the polytope.

Consider some independent set of integer vectors $f^{1}, f^{2}, \ldots, f^{k}$ lying in the cone
where $k=\dim A$.

Since we consider the projective version of dimension, we are working in projective space and after
an addition of vector $c \cdot \vec{1} = (c, \ldots, c)$
any vector in the cone remains in the cone.
So we can agree that all coordinates of vectors $f^1, \ldots, f^k$ are non-positive.

For each coordinate $i$ of the basis vectors $f^1, \ldots, f^k$ consider the tuple $f_i = (f^{1}_i, \ldots, f^{k}_i)$ of all $i$-th
coordinates of vectors in the basis. Note that $a = |\{f_1, f_2, \ldots, f_n\}| \geq k+1$,
that is the number of different vectors among $f_i$ is at least $k+1$.
Indeed, add the vector $f^{0} = (-1, \ldots, -1)$ to the basis to get the basis of the cone in $\bb{R}^n$.
If the number of different tuples among $f_i$ is less than $k+1$ then we can consider the vector equality $ \sum_i c_i f^i = \vec{0}$
for $c_i \in \mathbb{R}$
as a linear system on $c_0, c_1, \ldots, c_k$. This system has less than $k+1$ equalities and thus has nonzero solution.
This means that $f^0, f^1, \ldots, f^k$ are linearly dependent and we have a contradiction.

Next, consider $\eps_{1}, \ldots, \eps_{k}$ -- positive small enough
real numbers linearly independent over rationals and consider the vector
$$
f^{\prime} = \sum_{i=1}^{k} \eps_i f^i.
$$
``Small enough'' means that the absolute values of the coordinates of $f^{\prime}$ are less than $1$.
The number of different coordinates of this vector is equal to $a \geq k+1$.
Indeed, due to the linear independence of $\{\eps_i\}_{i=1}^{k}$ we have different sums for different $f_{i}$'s.
Let us enumerate the coordinates of $f^{\prime}$ in the increasing order: $b_{1} < \ldots < b_{a} \leq 0$.

Let us denote by $B_{j}$ for $j \in [a]$ the set of coordinates of $f^{\prime}$ with the values at most $b_{j}$,
that is $B_{j}= \{l \in [n] \mid f^{\prime}_{l} \leq b_{j}\}$. Note that $B_{1} \subset B_{2} \subset \ldots \subset B_{a}$.

\begin{claim}
For every $j$ and for every row $l$ columns in $B_{j}$ contain in $l$ either no stars,
or at least two stars in the table $A^{\ast}$
\end{claim}
\begin{proof}
The proof goes by induction on $j$.

For the base of induction consider the set $B_{1}$.
Note that the columns in $B_{1}$ are precisely the set of columns with the smallest
coordinates of $f^{\prime}$. Suppose that there is a star in the row $l$ and in the columns of $B_{1}$.
Let us add to each column $i$ of the matrix $A$ the $i$-th coordinate of $f^{\prime}$ (which is non-positive).
Since $f^{\prime}$ is a solution the resulting matrix should have at least two stars in row $l$. But the star
among the columns of $B_{1}$ has the value $b_1$
which is the smallest possible value of coordinates in the row $l$. Thus there should be one more coordinate with the same value
and this can appear only in columns of $B_{1}$ and it can only be a star of $A$.
Thus columns of $B_{1}$ have at least two stars in the row $l$.

For induction step assume that we have proved the claim for $B_{j-1}$ and
consider the set $B_{j}$. If row $l$ contains two stars in $B_{j-1}$ it also contains
two stars in $B_{j}$. Thus we can assume that row $l$ contains no stars in $B_{j-1}$.
Assume that there is a star in $B_{j}$. Again add coordinates of $f^{\prime}$ to the corresponding columns of $A$.
Since there are no stars in $B_{j-1}$ all corresponding coordinates of the row $l$ in these columns are positive
(recall that the coordinates of $f^{\prime}$ are less than $1$).
The star in $B_{j}$ has coordinate $b_{j}$ and this is the smallest possible value of coordinate.
Since $f^{\prime}$ is a solution to the system there should be one more coordinate with the same value
and this can be only the coordinate in $B_{j}$ and also the initial star of $A$.
Thus there are at least two stars in $B_{j}$ in the row $l$.
\end{proof}

Now we are ready to describe the sets of rows and columns
corresponding to the desired triangular form.
The size of this form will be $a$.
For the set $C_{a-i+1}$ we let $B_{i} \setminus B_{i-1}$ (note that $C_{a-i+1}$ is nonempty).
The choice of $R_{i}$ is straightforward:
we take all rows that have at least two stars in the set $C_{i}$
and no stars in $C_{i+1}, \ldots, C_{k}$.

Properties of the triangular form follows from the construction.
We only have to check that every row is in some $R_{i}$.
Consider arbitrary row $l$ and let $i$ be the smallest number such that $B_{i}$
contains a star of the row $l$. By the claim this star cannot be unique,
and since by the choice of $B_{i}$ there are no stars in row $l$ and in the columns of $B_{i-1}$, we have
$l \in R_{a-i+1}$.

\end{proof}

It is easy to see that the same argument works for the tropical linear systems over $\bb{Z}_{\infty}$:
we can ignore infinite coordinates of the solution we consider and infinite entries in the matrix do not affect the proof.
That is, given a solution $x \in \bb{Z}_{\infty}$ we remove from the matrix $A_x$
(see Preliminaries for the definition) all columns for which the corresponding coordinate of $x$
is infinite and denote the resulting matrix by $\widetilde{A}_{x}$. Consider the corresponding
table $\widetilde{A}_{x}^{\ast}$. It is not hard to see that the rows consisting of infinities
does not affect the maximal size of the block triangular form.
Note that infinities in other rows of the matrix can not become stars in $\widetilde{A}_{x}^{\ast}$ in the neighborhood
of $x$ and thus if we substitute them by large enough numbers neither the local dimension,
nor the block triangular forms of the maximal size change.

Almost the same argument works for min-plus linear systems $A \tp x = B \tp x$,
where $A, B$ are in $\bb{Z}^{m \times n}$ or $\bb{Z}_{\infty}^{m \times n}$.
Here we consider the joint matrix $D = \left(
                                         \begin{array}{c|c}
                                           A & B \\
                                         \end{array}
                                       \right)
$
and also consider the table $D^{\ast}$.
The block triangular form of size $d$ is now the row partition
$R_{1}, R_{2}, \ldots, R_{d}$, where some of the sets $R_{i}$ might be empty,
and the partition $C_{1}, \ldots, C_{d}$ of \{1, \ldots, n\}, where all $C_{i}$ are nonempty.
For a given set $C_{i}$ we associate the columns in $A$-part of $D$ with the corresponding numbers
and the columns in $B$-part of $D$ with the same numbers.
The partitions should satisfy the following properties:
\begin{enumerate}
\item for every $i$ each row in $R_{i}$ has at least one star in columns with numbers $C_{i}$ in $A$-part of $D$ and
at least one star in columns with numbers $C_{i}$ in  $B$-part of $D$;
\item if $1 \leq i < j \leq d$ then rows in $R_{i}$ have no stars in columns with numbers $C_{j}$ in both parts of $D^{\ast}$.
\end{enumerate}
The analog of Theorem~\ref{thm.dim_char} can be proven by a straightforward adaptation of the proof above.

Finally, the same construction works also for min-plus systems of inequalities over $\bb{Z}$ and $\bb{Z}_{\infty}$.
For the system $A \tp x \leq B \tp x$ we again consider the joint matrix $D = \left(
                                         \begin{array}{c|c}
                                           A & B \\
                                         \end{array}
                                       \right)
$ of left-hand side and right-hand side of the system, again consider $D^{\ast}$
and again consider similar partitions, but now we have different
requirements for partitions to be the block triangular form:
\begin{enumerate}
\item for every $i$ and each row $l$ in $R_{i}$ if there is a star in columns with numbers $C_{i}$ in $B$-part of $D$ in $l$
then there is a star in columns with numbers $C_{i}$ in $A$-part of $D$ in $l$;
\item if $1 \leq i < j \leq d$ then rows in $R_{i}$ have no stars in columns $C_{j}$ in $C^{\ast}$.
\end{enumerate}
Again, the combinatorial characterization is an easy adaptation of the proof of Theorem~\ref{thm.dim_char}.

\section{Computing the dimension of tropical and min-plus linear prevarieties is $\NP$-complete}
\label{sec.completeness}


Before proving the completeness result we prove the following technical lemma.

\begin{lemma} \label{lemma.size_bound}
If we are given a tropical linear system $A$ over $n$ variables the entries of which are nonnegative and  of value at most $M$,
then the maximal dimension of the tropical prevariety is achieved at some point
with all finite coordinates at most $(M+1)n$.
\end{lemma}
\begin{proof}
We have seen in Theorem~\ref{thm.dim_char} that the dimension of the tropical prevariety
in a given point depends only on the star-table in this point.
Given a star-table we consider a graph whose nodes are stars in the table
and two stars are connected if they are in the same column or in the same row.
We call this graph by star-graph.
We say that two columns of the table are connected if there are two stars in these columns
for which there is a path in star-graph between them.
Note that if there is a path there is always a path of length at most $2n$ ($n$ row-steps and $n$ column-steps).
If all columns are connected, then for each pair of solution's coordinates
there is a path in a star graph of length at most $2n$ connecting these two columns.
It is not hard to see that for each consecutive solution coordinates in this path their difference is at most $M$.

If not all columns are connected then there are several connected components.
We take one of them and reduce all coordinates in this component of the solution by the same number until new star appears in this set of columns.
It is easy to see that this star connects two different components. After that we increase all the coordinates we have just reduced by $1$.
Then on the place of a new star we have an entry which is by $1$ larger than the star-entries in the same row. Instead of star we put symbol $\circ$
in this entry. And from now on consider star-circle-graph.
Thus reducing components one by one and introducing new $\circ$-entries we get a connected graph.
Applying the argument for connected graphs we obtain the desired $(M+1)n$ upper bound.
\end{proof}

\begin{lemma} \label{lemma.dim_in_np}
$\tropdim \in \NP$ and $\tropdim_{\infty} \in \NP$.
\end{lemma}
\begin{proof}
As a certificate of an inequality $\dim A \geq k$ one can take a solution $x$
at which the local dimension is at least $k$, together with a block triangular form of $A_x = \{a_{ij} + x_{j}\}_{i,j}$
of size at least $k+1$ (see Theorem~\ref{thm.dim_char}).
By Lemma~\ref{lemma.size_bound} there is a solution needed with small enough coordinates.
It is easy to check in polynomial time that the given vector is a solution and that the given row and column partitions
indeed give a block triangular form of needed size.

The same proof works for $\tropdim_{\infty}$.
\end{proof}

To prove $\NP$-completeness we will give a reduction of $\vertex$ problem to our problem.

\begin{definition}
$\vertex$: given an undirected graph $G$ and a natural number $k$ decide whether there is a vertex cover of size at most $k$ in $G$,
that is whether there is a subset $K$ of vertices of $G$ of size at most $k$ such that each edge of $G$ has at least one end in $K$.
\end{definition}

Let $n$ be the number of vertices in $G$ and $m$ be the number of edges in $G$.
We make the following additional assumptions on $G$ and $k$:
\begin{enumerate}
\item $G$ is connected;
\item $k \leq 2n/3$.
\end{enumerate}
With these additional assumptions Vertex Cover problem is still $\NP$-complete (this follows from the standard proof of its completeness~\cite{garey_johnson79book}).

\begin{theorem} \label{thm.dim_completeness}
$\tropdim$ is $\NP$-complete.
\end{theorem}

\begin{proof}
Given a fixed graph $G$ we will construct the matrix $A$ of tropical linear system.
The matrix $A$ will have $(n+1)$ columns, $m$ rows and all its entries will be $0$ or $1$,
that is $A \in \zo^{m \times (n+1)}$.
Zero vector will be a solution of the tropical system $A$ and the global dimension will be attained on this solution.

Now we construct the matrix $A$. The first column of $A$ consists of zeros (and thus the first column of $A^{\ast}$
consists of stars).
All other columns are labeled by vertices of $G$ and rows are labeled by edges of $G$.
To the entry $(v,e)$ we put $0$ if and only if $v$ is one of the endpoints of $e$.
In particular, this means that every row of $A$ contains exactly $3$ zeros and one of them is in the first column.

\[
\begin{blockarray}{*{13}{p{3pt}}}
    &          &      &           &       & $v$      &           &          &      & $u$          &     &          &  \\
\begin{block}{p{3pt}(p{3pt}|*{3}{p{3pt}}|p{3pt}|*{3}{p{3pt}}|p{3pt}|*{3}{p{3pt}})}
    & 0        &      &           &       &      &           &          &      &           &     &          &  \\
    & 0        &      &           &       &      &           &          &      &           &     &          &  \\
    & $\vdots$ &      &           &       &      &           &          &      &           &     &          &  \\
$e$ & 0        & 1    & $\ldots$  & 1     & 0    &   1       & $\ldots$ &  1   &  0        &  1  & $\ldots$ & 1 \\
    & $\vdots$ &      &           &       &      &           &          &      &           &     &          &  \\
    & 0        &      &           &       &      &           &          &      &           &     &          &  \\
\end{block}
\end{blockarray}
\]

Now let us consider the zero solution to the tropical system $A$.
We are going to prove that the local dimension of the solution space in this solution is at least $n- k$
if and only if $G$ has a vertex cover of size $k$.
Here we consider projective dimension.

First consider the vertex cover $V_1 \subseteq V$ of the graph $G$.
Consider the set of columns $V_1$ in $A$ and add the first column to it.
It is not hard to see that this set of columns contains at least two zeros in any row:
the one in the first column and the other one in $V_1$, since $V_1$ is a vertex cover.
Thus all other columns lies in the prevariety and the codimension is at least $n-k$.

Now suppose that the dimension of the tropical prevariety is $n-d$.
Thus there is a block-triangular form of $A$ of size $n-d+1$ (see Theorem~\ref{thm.dim_char}).
Note that if the first column is in the set $C_i$
then for all other sets $C_{j}$ the sets $R_{j}$ are empty,
that is we can assume $i=n-d+1$.
Indeed, if this is not true, it is easy to see that there are no stars below the main diagonal blocks
except for the first column. Thus if there are two sets except $C_i$ such that their row-sets
are nonempty, we have a contradiction with the connectedness of $G$, and if there is only
one such set except $C_i$, it should contain all columns except the first one and the
size of this block triangular form is $2$, but we know that there is larger block triangular form
(recall that the size of vertex cover is at most $2n/3$).

Thus we have that the block triangular form has the following structure.
$R_{1}, \ldots, R_{n-d}$ are empty, $|C_{1}|=\ldots=|C_{n-d}|=1$ and thus
$|C_{n-d+1}|=d+1$ and $R_{n-d+1}=\{1, \ldots, m\}$.
Also the first column
is in $C_{n-d+1}$. It is easy to see that the set of all other columns in $C_{n-d+1}$
forms a vertex cover and thus $k \leq d$.

Now it is only left to show that the zero solution of the system $A$
achieves the maximal dimension in the prevariety.
Consider any solution $x$ of the system~\eqref{eq.tropical}.
Since we are in the projective
tropical space we can assume that $x_1=0$.
This means that the first column of the matrix
\begin{equation} \label{eq.other_solution}
B = \{a_{ij} + x_{j}\}_{i,j}
\end{equation}
is the same as for the zero solution.

\begin{claim}
For all $j = 1, \ldots, n$ we can assume that $x_{j} \geq 0$.
\end{claim}
\begin{proof}[Proof of the claim]
Assume on the contrary that $\alpha = \min_{j} x_j < 0$.
Let $C_{1} = \{j \mid x_{j} = \alpha\}$.
The set of columns $C_1$ corresponds to some set $V_{1} \subseteq V$ of vertices of the graph $G$
(note that $1 \notin C_{1}$).
There are two cases.

\paragraph{Case 1.} $V_1 \neq V$.
Since $G$ is connected there is an edge $e$ with one end in $V_1$ and the other end in $V \setminus V_1$.
Consider the row of the matrix~\eqref{eq.other_solution} corresponding to $e$.
It is clear that in one entry in this row we have $\alpha$ and in all others we have numbers greater than $\alpha$.
Thus this row in the table $B^{\ast}$ contain only one star and we have a contradiction.

\paragraph{Case 2.} $V_{1} = V$.
Then to obtain $B$ we have decreased all columns of $A$ by the same integer.
Thus there are exactly two stars in each row of $B^{\ast}$. And since the graph is connected
the maximal triangular form in this case has size two:
the first column with empty set of rows and all other columns with all rows.
Thus the dimension in this point of the prevariety is only $1$ which is less than for the zero solution.
\end{proof}

Now consider some column $j$ such that $x_{j} > 0$. It is not hard to see that all entries of the matrix $B$
in this column are greater than zero. And since the first column consists of zeros we have that in the column $j$
in $B^{\ast}$ there are no stars. Thus it is easy to describe how the table $B^{\ast}$ differs from $A^{\ast}$:
we just remove all stars in $A^{\ast}$ from the columns $j$ such that $x_{j} > 0$.
It is only left to show that the size of the largest triangular form for $A$ is at least the size of the largest
triangular form for $B$. For this consider the largest triangular form for $B$. Note that each column $j$
such that $x_{j} > 0$ should constitute the separate set $C_i$ with the empty set $R_{i}$ and note
that we can assume that all these sets are in the beginning of the list of $C_i$'s.
Consider the same system of $C_i$'s and $R_{i}$'s for the matrix $A$. It is easy to see
that this system is a triangular form for this matrix also.
Thus the maximal size of the triangular form for $A$ can be only greater than for $B$
and thus the dimension of the prevariety attains its maximum on the zero solution.
\end{proof}

As a corollary we have the following result.

\begin{corollary} \label{cor.dim_complete}
$\tropdim_{\infty}$ is $\NP$-complete.
\end{corollary}

\begin{proof}
The containment in $\NP$ was already proven in Lemma~\ref{lemma.dim_in_np}.
The completeness follows since there is a simple reduction from $\tropdim$ to $\tropdim_{\infty}$
given by Lemma~\ref{lemma.reduct_simple}.
\end{proof}

The results of this section easily generalizes to the case of min-plus linear system.
\begin{theorem}
Given a min-plus linear system and a natural number $k$, the problem of deciding whether the solution
space of the system is at least $k$ is $\NP$-complete.
\end{theorem}
Indeed, the analogs of Lemmas~\ref{lemma.size_bound} and~\ref{lemma.dim_in_np} can be proven in the same way.
To give a reduction from $\vertex$ consider the same matrix $A$ from the proof of Theorem~\ref{thm.dim_completeness}
and denote $A = \left(
                  \begin{array}{cc}
                    a_0, & A^{\prime} \\
                  \end{array}
                \right)
$, where $a_0$ is the first column of $A$.
Consider the min-plus system
$$\left(
                  \begin{array}{cc}
                    a_0 + \vec{1}, & A^{\prime} \\
                  \end{array}
                \right) \tp x = \left(
                            \begin{array}{cc}
                              a_0, & A^{\prime} + I \\
                            \end{array}
                          \right) \tp x,
                $$ where $I$ is the matrix of the corresponding size consisting of ones.
Then following the lines of the proof of Theorem~\ref{thm.dim_completeness}
it is easy to see that the size of the maximal block triangular form for this system in zero solution is equal to $n-k+1$ where $k$
is the size of minimal vertex cover, and that the size of block triangular form in all other solutions is at most the size of the block triangular form in zero
solution.
The proof of the first part is almost the same. For the second part we again consider arbitrary solution and
assume that $x_1 = 0$. If there is a negative coordinate in the solution, then we consider the smallest coordinate.
It is easy to see that there is a row such that there is a minimum in the corresponding column in the left-hand side
and there is no equal value in the right-hand side. Thus there are no negative coordinates in $x$ and the proof proceeds as before.
The analog of Corollary~\ref{cor.dim_complete} can be also proven in the same way.

Finally, we get the analogous result for min-plus systems of inequalities.
\begin{corollary}
Given a min-plus system of inequalities and a natural $k$, the problem
of deciding whether the solution space has dimension at least $k$ is $\NP$-complete (both over $\bb{Z}$ and $\bb{Z}_{\infty}$).
\end{corollary}
\begin{proof}
The containment in $\NP$ can be proven in the same way.
The completeness follows from Lemma~\ref{lemma.inequalities}.
\end{proof}

{\bf Acknowledgements}. The first author is grateful to Max-Planck Institut f\"ur
Mathematik, Bonn for its hospitality during the work on this paper.

The question on the complexity of equivalence of min-plus linear
prevarieties was posed by Vladimir Voevodsky which encouraged the
authors to study interrelations between tropical and min-plus linear
prevarieties.

{\small
\bibliographystyle{abbrv}
\bibliography{bib/tropical}
}

\newpage

\appendix

\part*{Appendix}

\section{Direct proof of equivalence of $\tropsolv$ and $\tropsolv_{\infty}$}
\label{app.solv_direct}

In this section we prove that the solvability
problem for tropical linear systems over $\bb{Z}_{\infty}$ is equivalent to
the solvability problem for tropical linear systems over $\bb{Z}$.

Throughout this section we denote by $A + x$ the matrix $A$
to each column of which we add the corresponding coordinate of the vector $x$ (in the Preliminaries it was denoted by $A_x$ but this notation is not convenient in this section).

The reduction in one direction was already proven in Lemma~\ref{lemma.reduct_simple}. If the matrix over $\bb{Z}$
has a solution over $\bb{Z}_{\infty}$ then it also has a solution over $\bb Z$:
just substitute all infinities in the solution vector by large enough numbers.

Let us prove the reduction in the other direction.
Suppose we are given tropical linear system $A \in \bb{Z}_{\infty}^{m \times n}$

First of all note, that if in $A$ there is a row of infinities we can remove it without
changing the solvability. If there is a column of infinities in $A$ the system always has a solution.
Thus in what follows we can assume that none of this is the case.

Next, we can assume that all
non-infinity entries in $A$ are nonnegative and bounded by $M-1$.
It is proven in~\cite{grigoriev10system} that if there is a solution, then there is a solution
the non-infinity coordinates of which are bounded by $Mn$.
Let $\beta$ be (say) $100Mn$, $\alpha$ be $200Mn$ and $\gamma$ be $300Mn$.

For all $i \in [n]$ consider the matrix
$$
A_i =
\left(
  \begin{array}{cc}
    A^{\infty \to \alpha}_{i} & A^{\infty \to \alpha} \\
    B & C \\
  \end{array}
\right)
$$
of size $(m+n-1) \times (n+n-1)$, where $A^{\infty \to \alpha}$ is obtained from $A$
by substituting all infinities by $\alpha$,
$A^{\infty \to \alpha}_{i}$ is obtained from $A^{\infty \to \alpha}$ by removing the $i$th
column. $B$ is the $(n-1) \times (n-1)$ matrix with $-\beta$ on the main diagonal and $\gamma$
everywhere else and $C$ is the matrix with the $i$th column consisting of zeros and all other columns
consisting of $\gamma$.

\begin{claim}
$A$ has a solution iff there exists $i \in [n]$ such that $A_i$ has a solution.
\end{claim}

\begin{proof}

Assume first that there is a solution $x = (x_1, \ldots, x_{n})$ to $A$.
We know that  there is another solution to $A$ such that all
non-infinity coordinates of the solution are non-negative and bounded by $Mn$.
We also know that there is at least one finite coordinate in the solution,
let it be the coordinate with the number $i$ and consider the matrix $A_i$.

Consider the following solution vector $(y_1, \ldots, y_{n-1}, z_{1}, \ldots, z_{n})$
of $A_{i}$. Let
$$
z_{j} = \begin{cases}
x_{j}, & x_j \neq \infty\\
x_{i} + \beta, & x_{j} = \infty\\
\end{cases}
$$
(in particular, $z_{i} = x_{i}$).
For each $j \in [n-1]$ let
$$
y_{j} = x_{i} + \beta.
$$
Note that for $j$ such that $x_{j} = \infty$ we have $y_{j} = z_{j}$.

It is clear that in each of the last $(n-1)$ rows we will have two minimums:
one in the zero entry of $C$ and the other in the $-\beta$-entry of $B$.
Consider now some row $l$ among the first $m$ rows. If the minimum for the solution $x$ of $A$ in this
row was finite, then it will be the same in $A_{i}$. Indeed, all finite values of $A + x$ in the row $l$
remain the same in $A_i + (y,z)$ and all other entries are at least $\beta - M$ which is much greater than any finite entry.
If, on the other hand, the minimum for the solution $x$ of $A$ in row $l$ was infinite than we have that for any column $j$
of $A$ either the entry $(l,j)$ is infinite, or the coordinate $x_{j}$ is infinite.
This means that in the $(A^{\infty \to \alpha} + z)$-part of the matrix $A_{i} + (y,z)$ all entries in row $l$ are
either at least $\alpha$, or approximately equal to $\beta$ (that is, differ from $\beta$ by at most $Mn$).
Note also that there is at least one entry which is approximately
equal to $\beta$: $A$ does not contain a row of infinities.
Now note that in the $(A^{\infty \to \alpha}_{i} + y)$-part of $A_{i} + (y,z)$ entries in the row $l$ are also either at least $\alpha$
or approximately equal to $\beta$. And finally note that for each $\beta$-entry in $A^{\infty \to \alpha}+z$ there is
an entry with the same value in $A^{\infty \to \alpha}_{i}+y$ (the one corresponding to the same column of $A$;
note that for $i$th column of $A^{\infty \to \alpha}$ the value in $l$ should be at least $\alpha$ since $x_{i}$ is finite).
So if we consider the smallest entry in the row $l$ it will be approximately equal to $\beta$ and we will have two minimums.

Assume now that there is a solution of $A_{i}$ for some $i$.
Assume first additionally that all minimums in the last $(n-1)$ rows are in the entries with values $0$ and $(-\beta)$.
Consider the smallest $z$-coordinate $j$ of the solution.
Let for simplicity $z_{j} = 0$. In particular, $z_j \leq z_i$ and thus $z_j$
at least by $\beta$ smaller than all $y$-coordinates of the solution (due to the last $n-1$ rows).

Next we construct the set of columns $J \subseteq [n]$ step by step. The plan is that
to get the solution of $A$ we will leave the coordinates from $J$ in $z$ as they are and will make all other coordinates infinite.
We start with $J = \{j\}$ (the smallest coordinate). Thus $J$ is already nonempty. During all the procedure we will have
that $z_{l}$ for all $l \in J$ is at most $Mn$.
To update $J$ we consider some column $l \in J$ and consider some small entry $(r, l)$ (at most $M$) of the matrix $A^{\infty \to \alpha}$ in this column.
For our solution $(y,z)$ there are minimums in the row $r$ of $A_i + (y,z)$ and the values in these minimums are at most the value in $(r,l)$.
Thus the values in these minimums are at most $a_{r,l} + z_{l}$ and thus they are in the $z$-part of the matrix.
We add the columns corresponding to these minimums to $J$. The process stops when we cannot increase the set $J$.
Since there are only $n$ columns under consideration the values of coordinates in $J$ cannot reach values greater than $Mn$.

Now we let $x_l = z_l$ if $l \in J$ and $x_l = \infty$ otherwise.
Let us prove that $x$ is a solution of $A$. Consider some row $r$.
If the minimum in $r$ was attained on the columns from $J$, it will still be attained there
(note that both minimums get to $J$).
Suppose now that the minimum was attained on the columns outside of $J$.
Then the corresponding entries of $x$ are infinities and we have to prove that
all entries of the row $r$ in the matrix $A + x$ are infinite.
Suppose that there is a finite entry in this row. Then at first its column is in $J$
and also the corresponding entry of $A^{\infty \to \alpha}$ is at most $M$.
But then we could update the set $J$ considering this column and this entry in the row $r$.
This means that the columns where the minimum in the row $r$ is attained should be in $J$ and we have a contradiction.
Thus all entries in the row $r$ in the matrix $A + x$ are infinities and we are done.

We have considered the special case in which all minimums in the last $(n-1)$ rows of $A_{i}$ are in $0$-entries and $(-\beta)$-entries.
Suppose now that we have a solution with the minimum in one of the last rows situated in $\gamma$-entry.
Then it should have a minimum in a $\gamma$-entry of $C$ (if there is no one in it, then the matrix consisting of $B$ and zero column has a
solution; there is only one solution for such matrix and this solution does not contain minimums in $\gamma$-entries).
Consider the smallest coordinate $j$ of $z$. Then there is also $\gamma$-minimum in the column corresponding to $z_j$
(note that $z_i$ cannot be the smallest one, it is at least by $\gamma$ greater than the one with the minimum).
Let us assume for simplicity that $z_{j} = 0$. Then we have that $z_i \geq \gamma$.
This in its turn means that for all $c \in [n-1]$ we have $y_{c} \geq \gamma + \beta$.
Indeed, all entries of $A_i + (y,z)$ in $C$ are at least $\gamma$.
If there is $y_c < \gamma + \beta$ consider the smallest such coordinate. The row with $-\beta$ in
the corresponding column has a single minimum.

Next we again construct the set $J$ along the same lines starting from $J = \{j\}$.
The construction is the same, but note that now not only the columns of $y$ cannot get into $J$,
but also the column corresponding to $z_i$ cannot get there.

Next again we let $x_l = z_l$ for $l \in J$ and we let $x_l = \infty$ otherwise.
Note that now $z_i = \infty$. By the same argument $x$ is a solution of $A$.

\end{proof}

It is easy to reduce the question about solvability of at least one of the several systems $A_{1}, \ldots, A_n$
to the question of solvability of a single system. Just consider large enough $\delta$ and the block matrix with the
the matrices $A_{i}$ on the diagonal and $\delta + A_{i}$ matrix everywhere outside the diagonal in the block-column $i$,
where $\delta + A_{i}$ means that we add $\delta$ to all entries of $A_{i}$.

Thus we have polynomial time $m$-reductions between the problems under consideration.

This proof can be generalized to min-plus linear systems also.
One direction is again easy.
For the other direction we apply the similar matrix construction,
but if the column $i$ we have chosen is from the left-hand side we add the ``copies'' of only left columns
and we add them to the right-hand side. Matrices $B$ and $C$ are similar.
The proof of the reduction follows the same lines.

\section{Direct proof of equivalence between $\tropsolv$ and $\tropimpl$ over $\bb{Z}$ and $\bb{Z}_{\infty}$}
\label{app.equivalence}

We start this section with a direct proof of the following theorem.

\begin{theorem} \label{thm.direct_equiv}
$\tropimpl \leq_{T} \tropsolv$.
\end{theorem}

\begin{proof}
Suppose we are given a tropical linear system $A$ of size $m \times n$ and
the tropical lineal equation $l$ with $n$ variables
and we want to check whether the system $A$ implies $l$.
We can assume that all entries of $A$ and $l$ are nonnegative and bounded by some $M$.

First check whether $A$ has a solution. If not, then $A$ implies $l$ and we are done.
If there is a solution to $A$ check whether there is a solution to the joint system
$A \cup \{l\}$. If there is no solution to this system we are also done: clearly $A$ does not imply $l$.

Thus from now on we can assume that there is a common solution to $A$ and $l$.
We also can assume without loss of generality that all the coefficients of $l$ are zeros.

For each pair $(i,j) \in [n] \times [n]$ such that $i \neq j$ consider the system $B_{ij} \in \zo^{n \times n}$
of the following form:
$$
B_{ij} = - \frac{1}{3Mn} \bordermatrix{
&   &   & &      & & j & i\cr
& 0 & 0 & 0&\cdots&0 & 0 & 0\cr
& 1 & 1 & 1&\cdots&1 & 1 & 0\cr
& 0 & 1 & 1&\cdots&1 & 1 & 0\cr
& 0 & 0 & 1&\cdots&1 & 1 & 0\cr
& \vdots & \vdots& \vdots& \ddots & \vdots & \vdots & \vdots\cr
& 0 & 0 & 0&\cdots&1 & 1 & 0\cr
& 0 & 0 & 0&\cdots&0 & 1 & 0\cr},
$$
where we have rearranged columns to clarify the picture.
That is, lower left corner of $B_{ij}$ of size $(n-1) \times (n-1)$ is upper triangular.
Note that columns other than $i$ and $j$ can be enumerated in an arbitrary way.
This is not important so we choose some enumeration and consider some fixed matrix $B_{ij}$
for each pair $(i,j)$. Note also that the first row of each $B_{ij}$ is exactly equation $l$.
It is easy to see that each tropical linear system $B_{ij}$ has no solution.

Consider also tropical systems $A_{ij}$ for each pair $(i,j)$ consisting of all
equations of $A$ and of all equation of $B_{ij}$ except $l$ (that is, except the first one).

\begin{claim}
Equation $l$ follows from the system $A$ iff all systems $A_{ij}$ have no solutions.
\end{claim}
\begin{proof}[Proof of the claim]
First let us assume that there is a solution for some system $A_{ij}$.
Then this solution does not satisfy $l$ since $B_{ij}$ has no solutions.
Thus we have a solution of $A$ which does not satisfy $l$ and $A$ does not imply $l$.

Now suppose that there is a solution of $A$ which does not satisfy $l$.
Recall that the tropical prevariety of $A$ forms the connected set of polytopes in projective space.
We additionally can prove the following.
\begin{claim}
The tropical prevariety can be represented as a union of polytopes in projective space
such that the vertices of these polytopes are nonnegative integer vectors which coordinates are bounded by $Mn$.
\end{claim}
Note that this claim provides another proof of Lemma~\ref{lemma.size_bound}.
\begin{proof}
Consider arbitrary solution $(x_1, \ldots, x_{n})$
of the tropical linear system $A$ and consider the matrix $A_{x}$ (see Preliminaries for the definition).
For this matrix we consider the table $A_{x}^{\ast}$ defined in Section~\ref{sec.characterization}:
the table has the same size as $A_x$, has $\ast$ symbols in the row-minimum entries of $A_x$ and has no other symbols in it.
The fact that $x$ is a solution is equivalent to the fact that each row of $A_x^{\ast}$ has at least two stars in it.

Now we define our polytopes.
Consider arbitrary subset $P \subseteq [m]\times [n]$, such that for any $i \in [m]$ there are at least two different elements
$j,k \in [n]$ such that $(i,j), (i,k) \in P$, and consider the set $X_P$ of solutions $x$ of
the tropical linear system $A$ such that for any $(i,j) \in P$ the table $A_x^{\ast}$ contains
a star in the entry $(i,j)$ (thus we allow some additional stars in $A_x^{\ast}$, but we require that the entries from $P$
should contain stars).
It is easy to see that any solution lies in some set $X_P$.

Note that for each row $r$ the restriction that there are stars in certain entries
can be represented as a system of linear equalities and inequalities (in the projective space).
Thus any set $X_P$ is a polytope.
Thus we have constructed the system of polytopes union of which is equal to the tropical prevariety.

Now when we have defined the set of polytopes we are ready to prove
the bound on the coordinates of vertices.
For this with star-table we associate a star-graph: its vertices are $\ast$-symbols of $A_x^{\ast}$ and we draw an edge between two
nodes if the corresponding $\ast$-symbols are either in the same column, or in the same row of $A_x^{\ast}$.

We claim that if $x$ is a vertex of some polytope $X_P$ then the star graph corresponding to $x$
is connected. Indeed, suppose it is not connected. Then there are at least two connected components
and each component can be characterized by the set of columns containing it.
Consider one of the component. Note that there is small enough $\eps>0$ such that if we add the same number
between $-\eps$ and $\eps$ to each column of this component the star table will not change and thus
this interval lies in $X_P$. Thus we have an interval in $X_P$ containing $x$ and thus $x$ is not a vertex.

Now we are ready to prove the bound on the coordinates of the vertices.
Note that each two solution coordinates which have stars in the same row
differs by at most $M$. Since for the vertices the star graph is connected,
any pair of coordinates can be connected by a path in star graph of length at most $2n$ with
at most $n$ row-edges. Thus the difference between any two coordinates is at most $Mn$.
Since we work in the projective space we can assume that the smallest coordinate is zero.
Finally, since the entries of the matrix are integer, the coordinates of the vertices are also integer.
\end{proof}

In what follows we need two properties of the constructed polytopes:
first, that each polytope has a vertex (or equivalently, no polytope contains a line),
second, the intersection of two polytopes under consideration is also a polytope of the same form.

For the first property suppose one of the polytopes $X_P$ contains a line $x(t) = \vec{a} + t\vec{b}$.
Since we are in the projective space we can assume that all coordinates of $\vec{b}$
are nonnegative and there is at least one zero coordinate. Then if there is an element $(i,j) \in P$ such
that $b_j = 0$ then note that there is no star in $(i,j)$ entry of $A_{x(t)}^{\ast}$ for small enough negative $t$.
If on the other hand, there is an element $(i,j) \in P$ such that $b_{j} > 0$ then there is no star in $A_{x(t)}^{\ast}$
for large enough $t$.

For the second property note that $X_{P_1} \cap X_{P_{2}} = X_{P_1 \cup P_{2}}$.
The proof follows straightforwardly from the definition of $X_P$.

Now we have that the solution prevariety is a connected set of polytopes those vertices have integer coordinates
between $0$ and $Mn$.
A solution of $A$ satisfying $l$ can be chosen to be a vertex of a polytope of the system $A \cup \{l\}$ and thus
to have coordinates bounded by $Mn$.
Note also that the solution of $A$ not satisfying $l$ can be chosen to have nonnegative integer coordinates
of size at most $Mn+1$.
Indeed, this solution lies in some polytope $X_P$ corresponding to $A$.
If this polytope does not intersect the tropical prevariety of $l$,
then we can just choose the vertex of this polytope. If on the other hand the polytope intersects with some polytope
for $l$ we can consider a polytope of the system $A \cup \{l\}$ lying in $X_P$ with the maximal number of stars in the row $l$.
This polytope has some vertex $x$ and
there is an interval lying in $X_P$ with one end in $x$ and the other end not satisfying $l$. Since we are
in the projective space we can assume that the vector corresponding to this interval starting in $x$ has nonnegative coordinates
and some of the coordinates are zeros. This means that we can add some positive numbers to the coordinates of $x$
and still stay in $X_{P}$ and this means that no stars essential to $X_P$ disappears after this operation.
But on the other hand some stars disappear from the row $l$ and note that on this interval some of the stars are not presented
in all points except $x$. Indeed, otherwise all of them are among increased components and before they disappear
there is a point where some new star appears and this contradicts the choice of the polytope for $A \cup \{l\}$.
But then since the coordinates of $x$ are integers we can add $1$ to all the coordinates that increase on the interval and still remain in $X_{P}$,
but this new point will not be a solution of $l$.

Thus both solutions has bounded integer coordinates and moreover we can connect
them by a piecewise linear path in the projective space each node of which has bounded integer coordinates.
Indeed we can connect them by a path going through our polytopes and for the nodes of the path we can choose the vertices
of polytopes: if we need to pass from one polytope to another we can do it through the node of their intersection
(recall that it is also the polytope of the form $X_P$).
Finally note that we can choose one segment of the path on which the solution of $A$ stops to be a solution of $l$.

Thus there is a line segment $\alpha$ lying in the tropical prevariety of $A$ and connecting two solutions of $A$
with integer coordinates, one satisfying $l$ and the other not satisfying $l$.
Moreover we can assume that only one point of $\alpha$ satisfies $l$.
The segment line $\alpha$ can be parameterized as $a + t b$ where $t$ ranges from $0$ to $1$,
$b$ is an integer vector and $a$ is the solution satisfying $l$. Both coordinates of $a$ and $b$ are bounded by $Mn$.
Let us denote $a = (a_{1}, a_{2}, \ldots, a_{n-1}, a_{n})$.
We can add the same integer to all coordinates of $a$ to make the smallest coordinate
of $a$ to be zero. Since all coordinates of $l$ are zeros we have that $a$ has at least two zero coordinates
and other coordinates of $a$ are nonnegative. Without loss of generality let $n$ be the smallest coordinate
of $b$ among all coordinates where $a$ is zero and let $n-1$ be the second smallest coordinate of $b$ where $a$ is zero.
Thus $a = (a_{1}, \ldots, a_{n-2}, 0, 0)$. We can also add a number to all coordinates of $b$ in such a way that $b_n = 0$.
All other coordinates of $b$ corresponding to the zero coordinates of $a$ are positive since $a$ is the unique solution of $l$
on the $\alpha$.
Now consider the system $A_{n,n-1}$ and consider the point of $\alpha$ corresponding to $t = 1/3Mnb_{n-1}$. Let us denote this point by $c$.
First, it is clear that $c$ satisfies $A$. Note also that the $n$-th coordinate of $c$ is zero, the $(n-1)$-th
coordinate of $c$ is $1/3Mn$ and all other coordinates of $c$ are at least $\min\{1/3Mn, 1 - 1/3\}$ (the first value corresponds to coordinates with $a_i=0$ and the second corresponds to nonzero coordinates of $a$). Thus we have that $c$ also satisfies all equations of $B_{n,n-1}$ except $l$
and thus $c$ satisfies $A_{n,n-1}$.
\end{proof}

By the claim above to check whether $A$ implies $l$ it is enough to check whether $O(n^2)$ tropical linear systems
have solutions
from which the theorem follows (note that formally the systems has rational entries, but we can multiply
them by $3Mn$).

\end{proof}

Next we prove the following theorem.
\begin{theorem}
$\tropimpl_{\infty} \leq_{T} \tropsolv_{\infty}$.
\end{theorem}

The proof remains almost the same as before but
to get the interval $\alpha$ we need two ends of it
to have infinities in the same coordinates. For this we note that since
the set of solutions of the system is closed under $\min$ operation,
there is a solution with the minimal (with respect to inclusion) set of infinite coordinates.
We call the set of coordinates infinite for all solutions by kernel.
Solutions having infinities only in kernel coordinates we call kernel solutions.
If the kernels of $A$ and $A \cup \{l\}$ are different, then the implication is not true.
If, on the other hand, the kernels are equal we can proceed as before.
Indeed, if there is a solution $x$ of $A$ which is not a solution of $l$,
then there is a kernel solution of $A$ which is not a solution of $l$.
Just consider some kernel solution $y$, add large enough constant to all its
coordinates and take a minimum with $x$. With this new solution in hand
we can repeat the argument of Theorem~\ref{thm.direct_equiv}.

Thus we have to check that the
kernel of infinities is the same for $A$ and for $A \cup \{l\}$ (if they are different, then $l$ does not follow from $A$).

For this let us first prove the following claim.

\begin{claim}
The kernels are different for $A$ and $A \cup \{l\}$ iff the kernel of $A \cup \{l\}$ contains
all finite coordinates of $l$ and the kernel of $A$ does not.
\end{claim}

\begin{proof}
Assume that the kernels of $A$ and $A \cup \{l\}$ are different.
Then the kernel of $A$ is strictly included in the kernel of $A \cup \{l\}$.
Consider some solution $x$ of $A$ which has infinities only in the kernel of $A$,
and consider some solution $y$ of $A \cup \{l\}$ which has infinities only in the kernel of $A \cup \{l\}$.
Then add sufficiently large number $C$ to $x$ and consider $z = \min\{x+C, y\}$.
If $C$ is large enough this vector differs from $y$ only in the coordinates belonging to the kernel of $A \cup \{l\}$
and not belonging to the kernel of $A$: in these coordinate $y$ is infinite and $z$ is very large but finite.
Note that $z$ is a solution of $A$, but is not a solution of $l$. Thus, after substituting coordinates
in the difference of kernels by arbitrary large numbers $y$ becomes not a solution of $l$. This
can only happen if the minimum in $l + y$ is infinite and the symmetric difference of the kernels
contain some finite coordinate of $l$.

The other direction is obvious.
\end{proof}

Thus, to check whether kernels of $A$ and $A \cup \{l\}$ are different it is enough to check whether
$A \cup \{l\}$ has a solution with a finite minimum in the row $l$ and whether $A$ has a solution
such that it has a finite coordinate among the finite coordinates of $l$.
The kernels are different iff the answer to the first question is `no' and the answer to the second question is `yes'.

\paragraph{Checking the solutions of $A \cup \{l\}$.}
First we show how to answer the first question.
Let us consider the matrix $A \cup \{l\}$.
Without loss of generality let the row $l$ have the form
$$
l = (\vec{c}, \infty, \ldots, \infty),
$$
where the coordinates of $\vec{c}$ are finite.
We apply Lemma~\ref{lemma.stars_restriction} to vector $\vec{c}$
and $C = 10Mn$. From this we get the system $D$ and consider the system $B = A \cup D$.

If $A \cup \{l\}$ has a solution with a finite minimum in $l$,
then there is such a solution with absolute value of coordinates bounded by $Mn$ (see~\cite{grigoriev10system}),
and this is also a solution of $B$.

On the other hand, any solution of $B$ has minimum in rows of $D$ only in $\vec{c}$-part
and thus is also a solution of $A \cup \{l\}$ and has a finite minimum in $l$.

Thus we have proven that $A \cup \{l\}$ has a solution with a finite minimum in $l$ iff
$B$ has a solution.

\paragraph{Checking the solutions of $A$.}
It remains to check whether $A$ has a solution such that it has finite coordinate among the
finite coordinates of $l$.

We will check for each finite coordinate of $l$ whether $A$ has a solution with the corresponding finite coordinate.
Consider some finite coordinate of $l$, without loss of generality assume that this is the first coordinate.

Consider the matrix
$$
B = \left(
      \begin{array}{c|cccc}
        \infty &   &        &        & \\
        \vdots &   &        &  A     & \\
        \infty &   &        &        & \\\hline
        0      & 0 & \infty & \ldots & \infty \\
      \end{array}
    \right).
$$
It is clear that $A$ has a solution with a finite first coordinate iff
$B$ has a solution with a finite minimum in the last row.
Now to find out whether $B$ has a solution with a finite minimum in the last row we apply
the argument of the previous paragraph.

\paragraph{Extension to min-plus linear systems.}
The argument of this section can be also extended to min-plus linear systems.

For the reduction over $\bb{Z}$ we define matrices $B_{ij}$ in the same way,
but now consider them only for columns on different sides of the system.
The proof follows the same lines.

For the kernel part of the proof everything remains the same except for the last argument.
There we have to specify to which part of the matrix we add the new column.
But it is easy to see that the proof works if we add the new column (the first column of $B$ above)
to the opposite side of the column we consider (the second column of $B$ above).

\end{document}